\documentclass[12pt]{article}
\usepackage[left=2.5cm,right=2.5cm,
    top=2.5cm,bottom=2cm,bindingoffset=0cm]{geometry}
\usepackage{amsthm}
\usepackage{graphicx,textcomp}
\usepackage{amsmath,bm,amsfonts,mathrsfs,amssymb}
\usepackage{hyperref}
\newtheorem{theorem}{Theorem}[section]

\newtheorem{prop}[theorem]{Proposition}

\usepackage{enumitem}   
\usepackage{amsmath, amsthm, amssymb, amstext}
\usepackage{mathrsfs}
\usepackage{array, amsfonts, mathrsfs}
\usepackage{latexsym, amsmath}

\usepackage{cancel}

\usepackage{xcolor}

\providecommand{\keywords}[1]
{
  \small	
  {\textit{Keywords: }} #1
}

\providecommand{\msc}[1]
{
  \small	
  {\textit{MSC: }} #1
}

\begin{document}
\title{On perturbations of the DN map of a disk causing changes of surface topology}
\author{D.V. Korikov\thanks{PDMI RAS} \thanks{ITMO University} \thanks{e-mail: thecakeisalie@list.ru, \ ORCID:\href{https://orcid.org/0000-0002-3212-5874}{0000-0002-3212-5874}} \thanks{This work was supported by the Ministry of Science and Higher Education of the Russian Federation (agreement 075-15-2025-344 dated 29/04/2025 for Saint Petersburg Leonhard Euler International Mathematical Institute at PDMI RAS).}}
\date{\today}
\maketitle
\begin{abstract}
We establish the degeneration of the Schottky double of a genus $1$ Riemann surface with boundary as its DN map tends to the DN map of the unit disk.
\end{abstract}

Let $(M,g)$ be a smooth orientable surface with boundary $\mathbb{T}$ (the unit circle) and the smooth metric $g$. Introduce the DN map of $(M,g)$ by $\Lambda:\, f\mapsto\partial_\nu u^f$, where $u^f$ is a harmonic extension of $f\in C^\infty(\mathbb{T};\mathbb{C})$ into $(M,g)$ and $\nu$ denotes the unit outvard normal on $\mathbb{T}=\partial M$. As is well-known \cite{LU}, the DN map $\Lambda$ determines $(M,g)$ up to a conformal diffeomorphism that does not move the points of $\Gamma$. 

It has recently been proved that the conformal class of a surface depends continuously (with respect to a Teichm\"uller metric) on its DN map provided that the surface topology is fixed \cite{BKTeich,Kor CAOT}. At the same time, the surface topology itself is unstable under small perturbations of the DN map \cite{Kor ZNS}. An example is a disk $M$ with a small handle glued in: as the handle size tends to 0, the DN map of $M$ tends (in the space $\mathbb{B}:=B(H^{1}(\mathbb{T};\mathbb{C});L_2(\mathbb{T};\mathbb{C}))$) to the DN map $|\partial_\varphi|$ ($\varphi={\rm arg}z$) of the (unperturbed) disk $\mathbb{D}:=\{z\in\mathbb{C} \ | \ |z|\le 1\}$. 

Based on the above example, one can suppose that in general case the topological defects (say, extra handles) of the perturbed surface are separated from its near-boundary part by short closed geodesics  as $\Lambda$ becomes close to $|\partial_\varphi|$ (thereby the  surface genus ``observed from the boundary'' is effectively lowered). However, any curve on $M$ can be inflated, without change of its DN map, by multiplying the metric by a conformal factor. Thus, the above hypotesis makes sense only for some canonical representative of the conformal class of metrics on $M$. It is natural to chose, as such representative, the metric of constant curvature $-1$ such that the boundary $\partial M=\mathbb{T}$ is geodesic in this metric. 

Such metric is constructed as follows. Recall that a conformal class of metrics and a choice of orientation on $M$ determine a complex structure on it (a biholomorphic sub-atlas of the smooth atlas on $M$ such that $\star dz=-idz$ holds in any holomorphic coordinates $z$, where $\star$ is the Hodge star operator associated with the metric and the orientation on $M$). Introduce the Schottky double $X$ of $M$ obtained by gluing two copies $M\times\{\pm\}$ of $M$ (endowed with opposite orientations) along their boundaries. Then $X$ is a Riemann surface without boundary endowed with the anti-holomorphic involution $\tau:\,(x,\pm)\mapsto(x,\mp)$ ($x\in M$) and $M$ can be embedded smoothly into $X$ by identifying it with one of its copy $M\times\{+\}$. If ${\rm gen}(M)>0$, then $X$ admits a unique conformal metric ${\rm h}=\rho|dz|^2$ of the Gaussian curvature $-1$ (a hyperbolic metric); thus, the involution $\tau$ is an isometry on $(X,{\rm h})$. Since the boundary $\partial M$ (embedded into $X$) coincides with the set of fixed points of the isometric map $\tau$, $\partial M$ becomes a geodesic curve on $(X,{\rm h})$ after suitable smooth reparametrization.

Denote by $l_1(M)$ the length of the shortest closed geodesic curve on $(M,{\rm h})$. In this note, we prove the following statement.
\begin{prop}
\label{main}
Let ${\rm gen}(M)=1$. Then $l_1(X)\to 0$ as $\Lambda\to|\partial_\varphi|$ in $\mathbb{B}$. 
\end{prop}
\begin{proof}
1) Introduce the {\it Hilbert transform} 
$$H=-\Lambda^{-1}\partial_\varphi$$ 
of $M$. Recall that any DN map $\Lambda$ is an order one pseudo-differental operator coinciding $|\partial_\varphi|$ up to a smoothing operator \cite{LeeU}. In addition, it is unbounded self-adjoint operator in $L_2(\mathbb{T};\mathbb{C})$. Then $H$ is a bounded operator in $L_2(\mathbb{T};\mathbb{C})$ and $iH$ is a bounded self-adjoint operator on the space $H^{1/2}(\mathbb{T};\mathbb{C})/\mathbb{C}$ endowed with the inner product 
$$(\cdot,\cdot)_{\Lambda}:=(\Lambda^{1/2}\cdot,\Lambda^{1/2}\cdot)_{L_2(\mathbb{T};\mathbb{C})}.$$ 
The spectrum of $H$ is described in Lemma 2.1, \cite{KorCJM}, see also Lemma 1, \cite{B}. The essential spectrum of $H$ consists of two eigenvalues $-i$, $+i$ of infinite multiplicity; the corresponding eigenspaces consist of the boundary traces of holomorphic/anti-holomorphic functions on $M$, respectively. The discrete spectrum spectrum of $H$ consists of zero (${\rm Ker}H=\mathbb{C}$) and two simple eigenvalues $\pm i\mu$ ($\mu_k\in(0,1)$). Let $(i\mu,\eta)$ and $(-i\mu,\overline{\eta})$ be the eigenpairs of $H$, where $\eta$, $\overline{\eta}$ are normalized with respect to the inner product $(\cdot,\cdot)_{\Lambda}$; since $H^2+I$ is a smoothing operator, we have $\eta\in C^\infty(\mathbb{T};\mathbb{C})$. The convergence $\Lambda\to|\partial_\varphi|$ in $\mathbb{B}$ implies 
$$H\to H_0 \text{ in } B(L_2(\mathbb{T};\mathbb{C});L_2(\mathbb{T};\mathbb{C})),$$ 
where $$H_0:=-|\partial_\varphi|^{-1}\partial_{\varphi}$$ is the Hilbert transform of $\mathbb{D}$ (the spectrum of $H$ is $\{0,i,-i\}$). Hence one has $\lambda_{\pm 1}\to\pm i$ and $\mu\to 1$.

2) Introduce the space $\mathcal{D}=\{\upsilon\in\Omega^1(M;\mathbb{R}) \ | \ d\upsilon=0, \ \delta\upsilon=0 \text{ in } M, \ \upsilon(\partial_\varphi)=0 \text{ on } \partial M\}$ of harmonic differentials on $M$ normal to $\partial M$. There is the following connection between elements of $\mathcal{D}$ and the space $H^0_{sym}(X;K):=\{w\in H^0(X;K)  \ | \ \overline{\tau^*w}=w\}$ of symmetric Abelian differentials on the double $X$ (see Lemma 2.2, \cite{KorCJM}): $\upsilon$ is an element of $\mathcal{D}$ if and only if $w:=(i-\star)\upsilon$ can be extended to the Abelian differential on $X$ by the symmetry $\overline{\tau^*w}=w$. Since $H^0(X;K)=H^0_{sym}(X;K)+iH^0_{sym}(X;K)$, we have ${\rm dim}_{\mathbb{R}}\mathcal{D}={\rm dim}_{\mathbb{R}}H^0_{sym}(X;K)={\rm dim}_{\mathbb{C}}H^0(X;K)={\rm gen}(X)=2$.

Let $f\in C^\infty(\mathbb{T};\mathbb{C})$. Then the Hodge orthogonal decomposition
\begin{equation}
\label{H decomp}
\star du^f=du^h-\star\upsilon
\end{equation} is valid, where $h\in\mathbb{C}^{\infty}(\mathbb{T};\mathbb{C})$ and $\upsilon\in\mathcal{D}$. Restricting the last equality to the boundary, one obtains 
$$\Lambda f=\partial_\varphi h-\upsilon(\nu), \quad -\partial_\varphi f=\Lambda h.$$
Integrating these equations, one arrives at 
\begin{equation}
\label{bound traces harm 1}
h=Hf, \qquad \upsilon(\nu)=-[\partial_\varphi\Lambda^{-1}\partial_\varphi+\Lambda]f=-\Lambda(H^2+I)f.
\end{equation}
Since ${\rm Ran}(H^2+I)=\{c_1\eta+c_2\overline{\eta}+c_3 \ | \ c_{1,2,3}\in\mathbb{C}\}$, one can assume that 
\begin{equation}
\label{bound traces harm 2}
f=c\eta+\overline{c\eta},
\end{equation} 
where $c\in\mathbb{C}$. Since ${\rm dim}_{\mathbb{R}}\mathcal{D}={\rm {\rm Ran}\Lambda(H^2+I)}=2$, each boundary trace $\upsilon(\nu)$ of $\upsilon\in\mathcal{D}$ admit representation (\ref{bound traces harm 1}), (\ref{bound traces harm 2}).

Let $\star du^{f'}=du^{h'}-\star\upsilon'$ be another Hodge decomposition, where $f'=c'\eta+\overline{c'\eta}$, $h'\in C^\infty(\mathbb{T};\mathbb{C})$, and $\upsilon'\in\mathcal{D}$. In view of (\ref{H decomp}), one has 
\begin{align*}
(f,f')_\Lambda=(\Lambda f,f')_{L_2(\mathbb{T};\mathbb{C})}=(du^f,du^{f'})_{L_2(M,T^*M)}=(\star du^f,\star du^{f'})_{L_2(M,T^*M)}=\\
=(du^h,du^{h'})_{L_2(M,T^*M)}+(\star \upsilon,\star \upsilon')_{L_2(M,T^*M)}=(h,h')_\Lambda+(\nu,\nu')_{L_2(M,T^*M)}=\\
=(Hf,Hf')_\Lambda+(\nu,\nu')_{L_2(M,T^*M)}=-(H^2 f,f')_\Lambda+(\nu,\nu')_{L_2(M,T^*M)}
\end{align*}
and
\begin{align*}
(\star\upsilon,\upsilon')_{L_2(M,T^*M)}=(du^{h}-\star du^{f},\nu')_{L_2(M,T^*M)}=(du^{h},\upsilon')_{L_2(M,T^*M)}-0=\\
=\int_M du^{h}\wedge\star\upsilon'=[d\star\nu'=0]=\int_M d(u^{h}\wedge\star\upsilon')=\int_{\partial M}h\star\upsilon'=\\
=(h,\upsilon'(\nu))_{L_2(\mathbb{T};\mathbb{C})}=(-Hf,-\Lambda(H^2+I)f')_{L_2(\mathbb{T};\mathbb{C})}=(Hf,(H^2+I)f')_\Lambda.
\end{align*}
Thus, the inner products between elements of $\mathcal{D}$, $\star\mathcal{D}$ on $M$ are expressed in terms of their boundary data as 
\begin{equation}
\label{bilinear forms via BD}
\begin{split}
(\nu,\nu')_{L_2(M,T^*M)}&=((H^2+I)f,f')_\Lambda=2(1-\mu^2)\Re(c\overline{c'}), \\ 
(\star\nu,\nu')_{L_2(M,T^*M)}&=(Hf,(H^2+I)f')_\Lambda=2(\mu^2-1)\mu\Im(c\overline{c'}).
\end{split}
\end{equation} 

Next, let us choose the canonical homology basis on the double $X$ as follows
\begin{equation}
\label{homobasis}
a_+,a_-:=\tau\circ a_+,b_+,b_-:=-\tau\circ b_+,
\end{equation}
where the curves representing the cycles $a_+,b_+$ belong to $M\subset X$. We introduce the dual basis $\upsilon_a,\upsilon_b$ in $\mathcal{D}$ obeying the conditions
$$\int_{q_+}\upsilon_e=\delta_{qe}, \qquad (q,e=a,b)$$
and represent their boundary data as 
$$\upsilon_\bullet(\nu):=-\Lambda(H^2+I)\big[c_\bullet\eta+\overline{c_\bullet\eta}\big] \qquad (\bullet=a,b).$$ 
Introduce the auxiliary period matrix $\mathfrak{B}$ with entries 
$$\mathfrak{B}_{qe}=\int_{q_+}\star\upsilon_e \qquad (q,e=a,b).$$ 
Then the Riemann bilinear relations 
$$\int_M \zeta\wedge\eta=\int_{a_+}\zeta\int_{b_+}\eta-\int_{b_+}\zeta\int_{a_+}\eta$$
(where $\zeta$ or $\eta$ is normal to $\partial M$) and formulas (\ref{bilinear forms via BD}) imply
\begin{align}
\label{periods of harm diff}
\begin{split}
1&=(\star\upsilon_a,\upsilon_b)_{L_2(M,T^*M)}=2\mu(\mu^2-1)\Im(c_a\overline{c_b}), \\
\mathfrak{B}_{bb}&=-\mathfrak{B}_{aa}=(\upsilon_a,\upsilon_b)_{L_2(M,T^*M)}=2(1-\mu^2)\Re(c_a\overline{c_b}),\\
\mathfrak{B}_{ab}&=-\|\upsilon_b\|^2_{L_2(M,T^*M)}=-2(1-\mu^2)|c_b|^2, \\
 \mathfrak{B}_{ba}&=\|\upsilon_a\|^2_{L_2(M,T^*M)}=2(1-\mu^2)|c_a|^2.
\end{split}
\end{align}
In particular, one has 
$${\rm det}\mathfrak{B}=\mu^{-2}.$$

3) Since homology basis (\ref{homobasis}) is symmetric with respect to the involution $\tau$, the dual basis in $H^0(X;K)$ is of the form $(\omega,\omega^\dag:=\overline{\tau^*\omega})$, where
\begin{equation}
\label{dubasis}
\int_{a_{\pm'}}\omega=\delta_{\pm',+}.
\end{equation}  
Let us represent $\omega$ in the form 
$$\omega=e_a w_a+e_b w_b,$$ 
where $e_a,e_b\in\mathbb{C}$ and the Abelian differentials $w_\bullet$ ($\bullet=a,b$) are obtained by the extension of $(i-\star)\upsilon_\bullet$ by symmetry $\overline{\tau^*w}=w$. From (\ref{periods of harm diff}), (\ref{homobasis}) and the obvious equality
$$\int_{l\circ\tau}\omega^\dag=\overline{\int_{l}\omega},$$ 
it follows that
\begin{equation}
\label{aux to dibasis}
\int_{a_\pm}\omega_a=-\mathfrak{B}_{aa}\pm i, \quad \int_{b_\pm}\omega_a=\mp \mathfrak{B}_{ba}, \quad \int_{a_\pm}\omega_b=-\mathfrak{B}_{ab}, \quad \int_{b_\pm}\omega_b=\mp\mathfrak{B}_{bb}+i.
\end{equation}
The coefficients
\begin{equation}
\label{coef of dibasis}
e_a=\frac{1}{2i}, \qquad e_b=\frac{i\mathfrak{B}_{11}-1}{2\mathfrak{B}_{12}}
\end{equation} 
are founded from (\ref{dubasis}) and (\ref{aux to dibasis}). The $b$-period matrix 
\begin{align*}
\mathbb{B}=\left(\begin{array}{cc}\int_{b_+}\omega & \int_{b_-}\omega\\
\int_{b_+}\omega^\dag & \int_{b_-}\omega^\dag\end{array}\right)=\left(\begin{array}{cc}
\gamma+i\delta & i\beta\\
i\beta & -\gamma+i\delta\end{array}\right) 
\end{align*}
of $X$ (equipped with homology basis (\ref{homobasis})) is founded from (\ref{coef of dibasis}) and (\ref{aux to dibasis}),
\begin{align}
\label{period matr asymp}
\mathbb{B}=-\frac{(\mu^{2}+1)i}{2\mu^2\mathfrak{B}_{ab}}\left(\begin{array}{cc}
1 & \frac{\mu^2-1}{\mu^{2}+1}\\
\frac{\mu^2-1}{\mu^{2}+1} & 1\end{array}\right)+\frac{\mathfrak{B}_{bb}}{\mu^2\mathfrak{B}_{ab}}\left(\begin{array}{cc}
1 & 0\\
0 & -1
\end{array}\right).
\end{align}
If the homology basis on $X$ is not specified, then the $b$-period matrix $\mathbb{B}$ of $X$ is determined up to the modular action of the group ${\rm Sp}(4;\mathbb{Z})$. If the homology basis on $X$ is constrained only by symmetry condition (\ref{homobasis}), then $\mathbb{B}$ is determined up to the modular action of a certain subgroup $G_{2,1}\subset{\rm Sp}(4;\mathbb{Z})$ described in \S 5, \cite{Giavedoni} (see also \cite{Silhol}). In Theorem 5.2, \cite{Giavedoni}, the fundamental domain of $G_{2,1}$ in the Siegel upper half-space $\mathcal{H}_2$ is described\footnote{The symmetric homology basis 
$$a_+,a'_-:=-\tau\circ a_+,b_+,b'_-:=\tau\circ b_+$$ 
used in \cite{Giavedoni} slightly differs from ours but fundamental domain remains the same. Indeed, the $b$-period matrices $\mathbb{B}$ and $\mathbb{B}'$ corresponding to these homology bases, differ only by the sign change for the off-diagonal entries $\beta\mapsto -\beta$. Thus $\mathbb{B}$ belong to domain (\ref{fundomain}) if and only if so does $\mathbb{B}'$.}; this domain is given by 
\begin{equation}
\label{fundomain}
\gamma\in[0,1/2], \qquad \delta^2\ge 1-\gamma^2+\beta^2 \,(\ge 3/4)\,, \qquad \delta^2>\beta^2.
\end{equation} 
Thus, one can always chose symmetric basis (\ref{homobasis}) is such a way that $\mathbb{B}$ belongs to set (\ref{fundomain}). Then the numbers $\mathfrak{B}_{ab}$ and $\mathfrak{B}_{bb}/\mathfrak{B}_{ab}$ are bounded as $\Lambda\to|\partial_\varphi|$.

4) Recall that any genus $2$ Riemann surface $X$ is hyperelliptic, i.e., it is biholomoprhic to a surface in $\mathbb{C}^2$ given by the equation 
\begin{equation}
\label{hypercu}
y^2=P(x)=x(x-1)(x-\lambda)(x-\mu)(x-\nu),
\end{equation}
where $P$ has no multiple roots. The values of $\lambda,\mu,\nu$ are related to the theta constants
\begin{equation}
\label{teta const}
e_{2a_1,2a_2,2b_1,2b_2}:=\vartheta\left[\begin{smallmatrix}\vec{a}\\ \vec{b}
\end{smallmatrix}\right](0|\mathbb{B})=\sum_{\vec{n}\in\mathbb{Z}^2}{\rm exp}\Big(\pi i(\vec{n}+\vec{a})^t\mathbb{B}(\vec{n}+\vec{a})+2\pi i(\vec{n}+\vec{a})^t\vec{b}\Big) \quad (a_i,b_j=0,\frac{1}{2})
\end{equation}
of $X$ via the formulas \cite{Gaudry,Wamelen}:
$$\lambda=\Big(\frac{e_{0000}\,e_{0010}}{e_{0011}\,e_{0001}}\Big)^2, \ \mu=\Big(\frac{e_{0010}\,e_{1100}}{e_{0001}\,e_{1111}}\Big)^2, \ \nu=\Big(\frac{e_{0000}\,e_{1100}}{e_{0011}\,e_{1111}}\Big)^2.$$
Now we assume that the surfaces, their DN-maps etc., depend on the parameter $r\in\mathbb{N}$ and $\Lambda_r\to|\partial_\varphi|$ as $r\to\infty$. Then there are the following two cases. 
\begin{enumerate}[label=\roman*)]
\item The set $\{\mathbb{B}_r\}_r$ is bounded as $\Lambda\to |\partial_\varphi|$. By passing to a sub-sequence, one can assume that $\mathbb{B}_r\to\mathbb{B}_{\infty}$ as $r\to\infty$. Then formulas (\ref{period matr asymp}) and the convergence $\mu\to 1$ imply that $\mathbb{B}_{\infty}={\rm diag}(B_1,B_2)$ is diagonal. Note that 
$$\vartheta\left[\begin{smallmatrix}\vec{a}\\ \vec{b}
\end{smallmatrix}\right](0|{\rm diag}(B_1,B_2))=\vartheta\left[\begin{smallmatrix}a_1\\ b_1
\end{smallmatrix}\right](0|B_1)\theta\left[\begin{smallmatrix}a_2\\ b_2
\end{smallmatrix}\right](0|B_2)$$
and $\theta[a_1;b_1](0|B_1)=0$ if and only if $a_1=b_1=1/2$. Due to these facts, one has $\mu,\nu\to\infty$ as $\mathbb{B}\to\mathbb{B}_{\infty}$. So, the (hyperelliptic) surface $X_r$ degenerates as $r\to+\infty$.

\item By passing to a sub-sequence, one can assume that $(\Im\mathbb{B})^{-1}\to 0$ as $r\to+\infty$. Then (\ref{teta const}) implies that $e_{2a_1,2a_2,2b_1,2b_2}\to\delta_{a_1,0}\delta_{a_2,0}$ whence $\lambda\to 1$.  
\end{enumerate}
It is worth noting that the asymptotics of the $b$-period matrix for families of degenerating surfaces were described in Corollaries 3.2 and 3.8, \cite{Fay2}; it has been shown that case i) corresponds to the pinching of the surface along a homologically trivial cycle while case ii) corresponds to the pinching along a homologically nontrivial one.

Since in both cases the surface $X_r$ (represented by curve (\ref{hypercu})) degenerates, the set $\{[X_r]\}_r$ of conformal classes of $X_r$ cannot be contained in any compact set of the moduli space $\mathcal{M}_2$.

5) If $l_1(X)\to 0$ does not hold, then there is the sub-sequence $\{X_r\}_r$ such that $l_1(X_r)\ge\epsilon>0$ for some $\epsilon$. Due to the Mumford compactness theorem \cite{Mum}, any set $K_\epsilon=\{[X]\in\mathcal{M}_2 \ | \ l_1(X)\ge\epsilon\}$ ($\epsilon>0$) is compact $\mathcal{M}_2$. However, since $\Lambda_r\to|\partial_\varphi|$ as $r\to\infty$, the above reasoning shows that $\{X_r\}_r$ cannot be contained in any $K_\epsilon$. This contradiction completes the proof of the convergence $l_1(X)\to 0$ as $\Lambda\to |\partial_\varphi|$.
\end{proof}

\begin{figure}[h!]
\center{\includegraphics[width=0.7\linewidth]{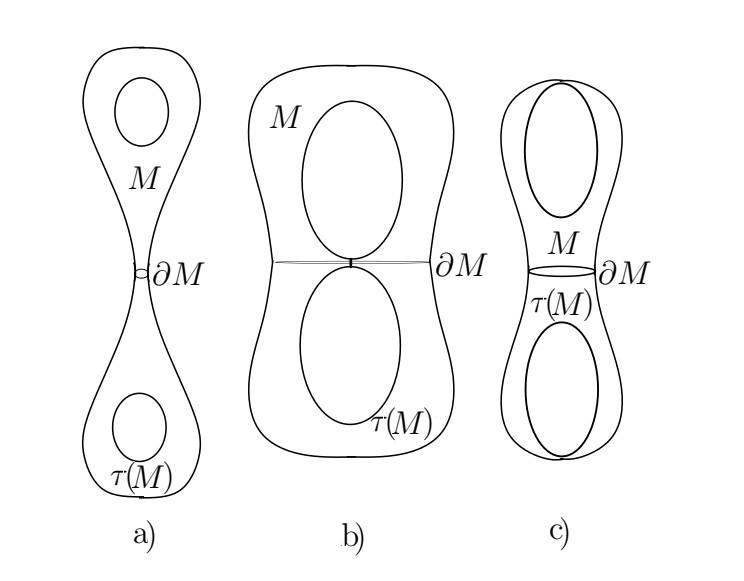}}
\caption{Different types of degeneration of $(X,{\rm h})$.}
\label{degpic}
\end{figure}

To exemplify the statement of Proposition \ref{main}, let us consider the case in which a small closed geodesic curve is the boundary $\partial M\subset X$ itself (see fig. \ref{degpic}, a)), i.e. $l(\partial M)\to 0$, where $l(\cdot)$ denotes the hyperbolic length of the curve on $X$. According to the collar lemma (see Lemma 13.6, \cite{Farb} and the proof of Corollary 4.4, \cite{Wang}), there is the metric neighborhood 
$$\{x\in X \ | \ {\rm dist}_{\rm h}(x,\partial M)\le L\}$$
biholomorphic to an annulus $\{z\in\mathbb{C} \ | \ {\rm log}|z|\in(-r,r)\}$, where 
$$L:={\rm sinh}^{-1}\Big(\frac{1}{{\rm sinh}(l(\partial M)/2)}\Big)\sim-{\rm log}(l(\partial M)), \qquad r=\frac{\pi^2}{l(\partial M)}.$$
Then $M$ is biholomorphic to a surface obtained by removing the set $$|z|\le e^{-r}=e^{-\pi^2/l(\partial M)}$$ from the unit disk $\mathbb{D}$ and then attaching a surface with a handle to the boundary $|z|=e^{-r}$. The surface families degenerating in this way were constructed in \cite{Kor ZNS}.


\msc{Primary: 35R30, 46J20; Secondary: 46J15, 30F15.}

\end{document}